\documentclass[preprint,12pt]{elsarticle}
\usepackage{amssymb,amsmath,amsthm}
\usepackage{booktabs}
\pdfoutput=1
\newtheorem{lemma}{Lemma}[section]
\newtheorem{definition}{Definition}[section]
\newtheorem{theorem}{Theorem}[section]

\newtheorem{proposition}{Proposition}[section]

\journal{}

\begin{document}

\begin{frontmatter}

\title{The tightly super 3-extra connectivity and 3-extra diagnosability of crossed cubes\tnoteref{label1}}
\tnotetext[label1]{ This work is supported by the National Science Foundation of China (61370001).}

\author{Shiying Wang\corref{cor}}
\cortext[cor]{Corresponding author.}
\ead{wangshiying@htu.edu.cn,shiying@sxu.edu.cn}

\author{Xiaolei Ma}
\ead{954631457@qq.com}

\address{School of Mathematics and Information Science, Henan Normal University, Xinxiang, Henan
              453007, PR China}
%\address[2]{Henan Engineering Laboratory for Big Data Statistical Analysis and Optimal Control}

\begin{abstract}
Many multiprocessor systems have interconnection networks as underlying topologies and an interconnection network
is usually represented by a graph where nodes represent processors and links represent communication links between processors.
In 2016, Zhang et al. proposed the $g$-extra diagnosability of $G$,
which restrains that every component of $G-S$ has at least $(g +1)$ vertices. As an important variant of the hypercube,
the $n$-dimensional crossed cube $CQ_{n}$ has many good properties.
In this paper, we prove that $CQ_{n}$ is tightly $(4n-9)$ super 3-extra connected for $n\geq 7$ and the 3-extra diagnosability of $CQ_{n}$ is $4n-6$ under the PMC model $(n\geq5)$ and MM$^*$ model $(n\geq7)$.

\end{abstract}


\end{frontmatter}

\section{Introduction}
\label{section1}

Many multiprocessor systems have interconnection networks (networks for short) as underlying topologies and a
network is usually represented by a graph where nodes represent processors and links represent communication links between processors. Some processors may be faulty when the system is in operation. The first step to deal with faults is to identify the faulty processors from the faultfree
ones. The identification process is called the diagnosis of the system.
A system $G$ is said to be $t$-diagnosable if all faulty processors can be identified without replacement, provided that the number of faults presented does not exceed $t$. The diagnosability $t(G)$ of $G$ is the maximum value of $t$ such that $G$ is $t$-diagnosable.

In order to discuss the connectivity and diagnosability in different situations, people put forward restricted connectivity and diagnosability in the system.
In 1996, F$\grave{a}$brega and Fiol  \cite{jf} introduced the $g$-extra connectivity of a system $G=(V(G),E(G))$, which is denoted by $\tilde{\kappa}^{(g)}(G)$.
A vertex subset $S\subseteq V (G)$ is called a $g$-extra vertex cut if $G-S$ is disconnected and every component of $G-S$ has at least $(g+1)$ vertices. $\tilde{\kappa}^{(g)}(G)$ is defined as the cardinality of a minimum $g$-extra vertex cut.
For a hypercube $Q_{n}$, Yang et al. \cite{ym} determined $\tilde{\kappa}^{(3)}(Q_{n})=4n-9$
for $n \geq 4$; For the folded hypercubes $FQ_{n}$, Chang et al. \cite{ct} determined $\tilde{\kappa}^{(3)}(FQ_{n})=4n-5$ for $n \geq 6$; Gu et al. studied the 3-extra connectivity of 3-ary $n$-cubes \cite{gh1} and $k$-ary $n$-cubes  \cite{gh2}.
For the star graph $S_{n}$ and the bubble-sort graph $B_{n}$, Li et al. \cite{lm} determined $\tilde{\kappa}^{(3)}(S_{n})=4n-10$ for $n \geq 4$, $\tilde{\kappa}^{(3)}(B_{n})=4n-12$
for $n \geq 6$. Chang et al. \cite{ch} determined the $n$-dimensional hypercube-like networks 3-extra connectivity, $\tilde{\kappa}^{(3)}(HL_{n})=4n-9$ for $n \geq 6$, and so on.

In 2016, Zhang et al. \cite{za} proposed the $g$-extra diagnosability of a system, which restrains that every fault-free component has at least $(g+1)$ fault-free vertices. They proved that the $g$-extra diagnosability of the $n$-dimensional hypercube under the PMC model and MM* model.
In 2016, Wang et al. \cite{sWang} studied
the $2$-extra diagnosability of the bubble-sort star graph $BS_{n}$ under the PMC model and MM$^*$ model. In 2017, Wang and  Yang \cite{Wang4} studied
the 2-good-neighbor (2-extra) diagnosability of alternating group graph networks under the PMC model and MM$^*$ model. In 2017, Ren and Wang \cite{Ren2}, studied the tightly super 2-extra connectivity and 2-extra diagnosability of locally twisted cubes.

Now, there are many topologies on the network.
Hypercube as an important network model has good properties such as lower diameter and node degree, high connectivity, regular, symmetry, and so on.
Efe and Member proposed crossed cube \cite{ef} by
changing the links between some nodes of hypercubes,
which have superior properties over hypercubes.
For example, its diameter is about half
of hypercube with the same dimension, which makes that the
communication speed between any two nodes is increased by
almost a half.

Several models of diagnosis have been studied in system level diagnosis.
Among these models, the most popular two models are the PMC and MM, which are proposed by Preparata et al. \cite{pp} and Maeng et al. \cite{mm}, respectively.
In the PMC model, only neighboring processors are allowed
to test each other.
In the MM model, a node tests its two neighbors, and then compares their responses.
Sengupta and Dahbura \cite{sd} suggested a special case of the MM model, namely the MM*
model, and each node must test its any pair of adjacent nodes in the MM*.

In this paper, we proved that
(1) $CQ_{n}$ is tightly $(4n-9)$ super 3-extra connected for $n\geq 7$;
(2) the 3-extra diagnosability of $CQ_{n}$ is $4n-6$ under the PMC model for $n\geq 5$;
(3) the 3-extra diagnosability of $CQ_{n}$ is $4n-6$ under the MM* model for $n\geq 7$.

\section{Preliminaries}
\label{section2}
\subsection{Notations}

A multiprocessor system is modeled as an undirected simple graph $G=(V,E)$, whose vertices (nodes) represent processors and edges (links) represent communication links.
The degree $d_{G}(v)$ of a vertex $v$ in $G$ is the number of neighbors of $v$ in $G$.
For a vertex $v$, $N_{G}(v)$ is the set of vertices adjacent to $v$ in $G$.
Let $S\subseteq V(G)$. $N_{G}(S)=\cup_{v\in S}N_{G}(v)\backslash S$ and $G[S]$ is the subgraph of $G$ induced by $S$.  A cycle with length $n$ is called an $n$-cycle.
We use $P= v_{1} v_{2}\cdots v_{n}$ to denote a path that begins with $v_{1}$ and ends with $v_{n}$. A path of the length $n$ is denoted by $n$-path.
A bipartite graph is one whose vertex set can be partitioned into two subsets $X$ and $Y$, so that each edge has one end in $X$ and one end in $Y$; such a partition $(X,Y)$ is called a bipartition of the graph. A complete bipartite graph is a simple bipartite graph with bipartition $(X,Y)$ in which each vertex of $X$ is joined to each vertex of $Y$; If $|X|=m$ and $|Y|=n$, such a graph is denoted by $K_{m,n}$.
The connectivity $\kappa(G)$ of a connected graph $G$ is the minimum number of vertices whose removal results in a disconnected graph or only one vertex left when $G$ is complete. Let $F_{1}$ and $F_{2}$ be two distinct subsets of $V$, and let the symmetric difference $F_{1}\vartriangle F_{2}=(F_{1}\setminus F_{2})\cup(F_{2}\setminus F_{1})$. For graph-theoretical terminology and notation not defined here we follow \cite{bj}.

A connected graph $G$ is super $g$-extra connected if every minimum $g$-extra cut $F$ of $G$ isolates one connected subgraph of order $g+1$. In addition, if $G-F$ has two components, one of which is the connected subgraph of order $g+1$, then $G$ is tightly $|F|$ super $g$-extra connected.

\subsection{The crossed cube $CQ_{n}$}

\begin{definition}(\cite{wm1})Let $R=\{(00,00),(10,10),(01,11),(11,01)\}$. Two digit binary strings $u=u_{1}u_{0}$ and $v=v_{1}v_{0}$ are pair related, denoted as $u\sim v$, if and only if $(u,v)\in R.$
\end{definition}

\begin{definition}(\cite{wm1}) \label{def} The vertex set of a crossed cube $CQ_{n}$ is $\{v_{n-1}v_{n-2} \cdots v_{0}: 0\leq i\leq n-1, v_{i}\in\{0,1\} \}$. Two vertices $u=u_{n-1}u_{n-2}\cdots u_{0}$ and $v=v_{n-1}v_{n-2}\cdots v_{0}$ are adjacent if and only if one of the following conditions is satisfied.\\
1. There exists an integer $l\ (1\leq l\leq n-1)$ such that \\
(1) $u_{n-1}u_{n-2}\cdots u_{l}=v_{n-1}v_{n-2}\cdots v_{l}$;\\(2) $u_{l-1}\neq v_{l-1}$;\\(3) if $l$ is even, $u_{l-2}=v_{l-2}$;\\(4) $u_{2i+1}u_{2i}\sim v_{2i+1}v_{2i}$, for $0\leq i<\lfloor\frac{l-1}{2}\rfloor$.\\
2. \\(1)$u_{n-1}\neq v_{n-1}$;\\(2)if $n$ is even, $u_{n-2}=v_{n-2}$;\\(3) $u_{2i+1}u_{2i}\sim v_{2i+1}v_{2i}$ for $0\leq i<\lfloor\frac{n-1}{2}\rfloor$.
\end{definition}

Let $n\geq 2$. We define two graphs $CQ^0_{n}$ and $CQ^1_{n}$ as follows. If $u=u_{n-2}u_{n-3}\cdots u_{0}\in  V(CQ_{n-1})$, then $u^0=0u_{n-2}u_{n-3}\cdots u_{0}\in  V(CQ^0_{n})$ and $u^1=1u_{n-2}u_{n-3}\cdots u_{0}\in  V(CQ^1_{n})$. If $uv\in  E(CQ_{n-1})$, then  $u^0v^0\in  E(CQ^0_{n})$ and  $u^1v^1\in  E(CQ^1_{n})$. Then $CQ^0_{n}\cong CQ_{n-1}$ and $CQ^1_{n}\cong CQ_{n-1}$. Define the edges between the vertices of $CQ_{n}^{0}$ and $CQ_{n}^{1}$ according to the following rules. \\
The vertex $u=0u_{n-2}u_{n-3}\cdots u_{0}\in V(CQ_{n}^{0})$ and the vertex
$v=1v_{n-2}v_{n-3}\cdots v_{0}\in V(CQ_{n}^{1})$ are adjacent
if and only if\\1. $u_{n-2}=v_{n-2}$ if $n$ is even;
\\2. $(u_{2i+1}u_{2i},v_{2i+1}v_{2i})\in R$, for $0\leq i<\lfloor\frac{n-1}{2}\rfloor$.

The edges between the vertices of $CQ_{n}^{0}$ and $CQ_{n}^{1}$ are said to be cross edges.

\begin{proposition}(\cite{wm1})\label{match} Let $CQ_{n}$ be the crossed cube. Then all cross edges of $CQ_{n}$ is a perfect matching.
\end{proposition}

By Proposition \ref{match},  $CQ_{n}$ can be recursively defined as follows.

\begin{definition}(\cite{wm1}) Define that $CQ_{1}\cong K_{2}$. For $n\geq2$, $CQ_{n}$ is obtained by $CQ_{n}^{0}$ and $CQ_{n}^{1}$, and a perfect matching between the vertices of $CQ_{n}^{0}$ and $CQ_{n}^{1}$
according to the following rules (see Fig.1):\\
The vertex $u=0u_{n-2}u_{n-3}\cdots u_{0}\in V(CQ_{n}^{0})$ and the vertex
$v=1v_{n-2}v_{n-3}\cdots v_{0}\in V(CQ_{n}^{1})$ are adjacent in $CQ_{n}$
if and only if\\1. $u_{n-2}=v_{n-2}$ if $n$ is even;
\\2. $(u_{2i+1}u_{2i},v_{2i+1}v_{2i})\in R$, for $0\leq i<\lfloor\frac{n-1}{2}\rfloor$.
\end{definition}

%TeXCAD Picture [CQ2,CQ3.pic]. Options:
%\grade{\on}
%\emlines{\off}
%\epic{\off}
%\beziermacro{\on}
%\reduce{\on}
%\snapping{\off}
%\quality{8.00}
%\graddiff{0.01}
%\snapasp{1}
%\zoom{8.0000}
\unitlength 1mm % = 2.85pt
\linethickness{0.4pt}
\ifx\plotpoint\undefined\newsavebox{\plotpoint}\fi % GNUPLOT compatibility
\begin{picture}(72.13,35.25)(0,0)
\put(8,6){\framebox(24,24)[]{}}
\put(46,6){\framebox(24,24)[]{}}
%\emline(46,30)(70,6)
\multiput(46,30)(.0337078652,-.0337078652){712}{\line(0,-1){.0337078652}}
%\end
%\emline(69.94,29.95)(46.01,6.02)
\multiput(69.94,29.95)(-.0337081594,-.0337081594){710}{\line(0,-1){.0337081594}}
%\end
\put(52.18,23.71){\line(1,0){11.45}}
\put(51.66,11.67){\line(1,0){12.56}}
\put(7.88,30){\circle*{1}}
\put(32,29.88){\circle*{1}}
\put(8,6.13){\circle*{1}}
\put(32.13,6){\circle*{1}}
\put(46.13,29.88){\circle*{1}}
\put(70,30.13){\circle*{1}}
\put(52.13,23.75){\circle*{1}}
\put(63.5,23.63){\circle*{1}}
\put(51.75,11.63){\circle*{1}}
\put(64.25,11.5){\circle*{1}}
\put(45.88,6.13){\circle*{1}}
\put(70.13,6){\circle*{1}}
\put(5,32){\small00}
\put(31,32){\small10}
\put(5,2){\small01}
\put(31,2){\small11}
\put(40,32){\small000}
\put(70,32){\small010}
\put(40,2){\small001}
\put(70,2){\small011}
\put(46,20){\small100}
\put(63.5,20){\small110}
\put(46,12.8){\small111}
\put(63.5,12.8){\small101}
\end{picture}

%TeXCAD Picture [jia.pic]. Options:
%\grade{\on}
%\emlines{\off}
%\epic{\off}
%\beziermacro{\on}
%\reduce{\on}
%\snapping{\off}
%\quality{8.00}
%\graddiff{0.01}
%\snapasp{1}
%\zoom{4.0000}
\unitlength 1mm % = 2.85pt
\linethickness{0.4pt}
\ifx\plotpoint\undefined\newsavebox{\plotpoint}\fi % GNUPLOT compatibility
\begin{picture}(74.9,46.56)(0,0)
\put(8.9,12.13){\framebox(25,25.25)[]{}}
\put(8.9,37.38){\line(1,-1){24.75}}
%\emline(33.9,37.38)(8.9,12.13)
\multiput(33.9,37.38)(-.0337381916,-.0340755735){741}{\line(0,-1){.0340755735}}
%\end
\put(48.4,12.38){\framebox(25.25,25.25)[]{}}
\put(48.4,37.63){\line(1,-1){25}}
\put(73.4,37.38){\line(-1,-1){24.75}}
\put(48.65,12.63){\line(-1,0){.25}}
\put(16.9,29.38){\line(1,0){9}}
\put(16.65,19.63){\line(1,0){10.25}}
\put(56.15,19.88){\line(1,0){10}}
\put(55.9,30.13){\line(1,0){10.5}}
\qbezier(15.8,20.1)(39.53,8.38)(66.15,19.63)
\qbezier(16.6,29.38)(36.28,41.01)(56.15,30.13)
\qbezier(26.2,29.63)(51.28,40.51)(66.15,29.88)
\put(26.4,19.63){\line(1,0){30}}
\put(8.8,38.5){\makebox(0,0)[rb]{\small0000}}
\put(26,38.5){\makebox(0,0)[lb]{\small0010}}
\put(8.8,11.5){\makebox(0,0)[rt]{\small0001}}
\put(32,11.5){\makebox(0,0)[lt]{\small0011}}
\put(16.8,28.5){\makebox(0,0)[rt]{\small0100}}
\put(26,28.5){\makebox(0,0)[lt]{\small0110}}
\put(16.8,21){\makebox(0,0)[rb]{\small0111}}
\put(26.65,20.5){\makebox(0,0)[lb]{\small0101}}
\put(55,38.5){\makebox(0,0)[rb]{\small1000}}
\put(73.5,38){\makebox(0,0)[lb]{\small1010}}
\put(50,11.5){\makebox(0,0)[rt]{\small1001}}
\put(73.7,11.5){\makebox(0,0)[lt]{\small1011}}
\put(56,29.5){\makebox(0,0)[rt]{\small1100}}
\put(65.8,29.5){\makebox(0,0)[lt]{\small1110}}
\put(55.9,21){\makebox(0,0)[rb]{\small1111}}
\put(66.4,21){\makebox(0,0)[lb]{\small1101}}
\put(8.9,12.2){\circle*{1}}
\put(33.9,12.13){\circle*{1}}
\put(8.9,37.63){\circle*{1}}
\put(33.9,37.63){\circle*{1}}
\put(16.65,29.5){\circle*{1}}
\put(26.15,29.5){\circle*{1}}
\put(16.4,19.8){\circle*{1}}
\put(26.9,19.63){\circle*{1}}
\put(55.6,30.38){\circle*{1}}
\put(65.9,30.13){\circle*{1}}
\put(55.9,19.88){\circle*{1}}
\put(66.15,19.88){\circle*{1}}
\put(48.65,12.38){\circle*{1}}
\put(73.65,12.38){\circle*{1}}
\put(73.4,37.63){\circle*{1}}
\put(48.4,37.6){\circle*{1}}
\put(34.06,12.25){\line(1,0){14.85}}
\qbezier(8.6,37.3)(33,46.56)(48.25,37.62)
\qbezier(33.87,37.37)(55,46.56)(73.62,37.5)
\qbezier(8.87,12.12)(40.19,2.25)(73.75,12.37)
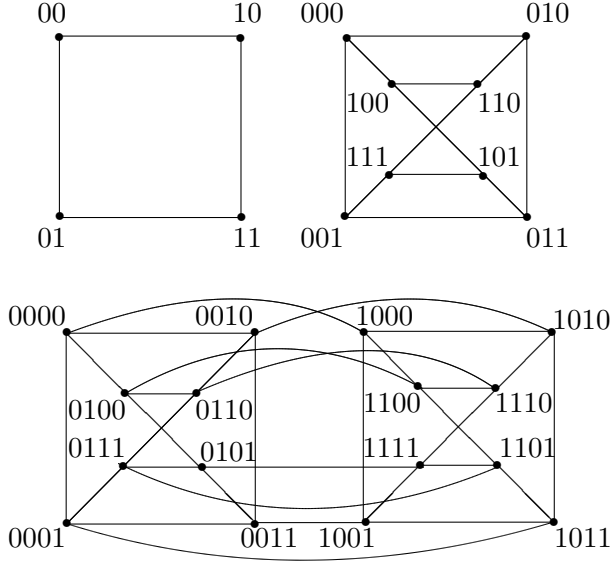
\put(15,0){Fig.1. $CQ_{2}$, $CQ_{3}$, and $CQ_{4}$.}
\end{picture}\\[1mm]

\section{The connectivity of crossed cubes}

\begin{lemma} (\cite{ef})\label{k} Let $CQ_{n}$ be the crossed cube. Then $\kappa(CQ_{n})=n$ for $n\geq1$.
\end{lemma}

\begin{lemma}(\cite{wm1})\label{2n-3} Let $CQ_{n}$ be the crossed cube and let $F\subseteq V(CQ_{n})$ $(n\geq3)$ with $n\leq |F|\leq 2n-3$. If $CQ_{n}-F$ is disconnected,
then $CQ_{n}-F$ has exactly two components, one of which is an isolated vertex.
\end{lemma}

\begin{lemma}(\cite{wm1})\label{3n-6} Let $CQ_{n}$ be the crossed cube and let $CQ_{n}$ be the crossed cube and let $F\subseteq V(CQ_{n})$ $(n\geq 5)$ with $2n-2\leq |F|\leq 3n-6$.  If $CQ_{n}-F$ is disconnected, then $CQ_{n}-F$ satisfies one of the following conditions:

(1) $CQ_{n}-F$ has two components, one of which is a $K_{2}$;

(2) $CQ_{n}-F$ has two components, one of which is an isolated vertex;

(3) $CQ_{n}-F$ has three components, two of which are isolated vertices.
\end{lemma}

\begin{lemma} \label{A} Let $CQ_{n}$ be the crossed cube and let $A=\{0\cdots0000,0\cdots0100,\\0\cdots0110,0\cdots0111\}$. If $F=N_{CQ_{n}}(A)$, then $|F|=4n-9$ and $CQ_{n}-(A\cup F)$ is connected for $n\geq4$.
\end{lemma}

\begin{proof} By the definition of the crossed cube, $CQ_{n}[A]$ is a 3-path. We proof this lemma by induction on $n$. In $CQ_{4}$ (see Fig. 1), $A=\{0000,0100,0110,0111\}$ and $F=\{0001,0010,0101,1000,1100,1110,1101\}$. It is easy to see that $|A|=4$, $|F|=7$ and $CQ_{4}-(A\cup F)$ is connected.
We can decompose $CQ_{n}$ along dimension $n-1$ into $CQ_{n}^{0}$ and $CQ_{n}^{1}$. Then both $CQ_{n}^{0}$ and $CQ_{n}^{1}$ are isomorphic to $CQ_{n-1}$. Let $F_{0}=F\cap V(CQ_{n}^{0})$ and $F_{1}=F\cap V(CQ_{n}^{1})$. Then $|F_{0}|+|F_{1}|=|F|$.
We assume that the lemma is true for $n-1$, i.e., if $F=N_{CQ_{n-1}}(A)$, then
$|F|=4(n-1)-9=4n-13$ and $CQ_{n-1}-(A\cup F)$ is connected. Now we proof that the lemma is also true for $n$ $(n\geq5)$.
Note that $A\in V(CQ_{n}^{0})$.
By the inductive hypothesis, we have $|F_{0}|=4n-13$ and $CQ_{n}^{0}-(A\cup F_{0})$ is connected. By Proposition \ref{match}, $|F_{1}|=4$. Thus, $|F|=|F_{0}|+|F_{1}|=4n-13+4=4n-9$. Now we prove that $CQ_{n}-(A\cup F)$ is connected for $n\geq5$.

By Lemma \ref{k}, $\kappa(CQ_{n}^{1})=n-1>4=|F_{1}|$ for $n\geq6$. Thus, $CQ_{n}^{1}-F_{1}$ is connected for $n\geq6$. We consider that $n=5$. Note that $A=\{00000,00100,00110,00111\}$. By Proposition \ref{match}, $F_{1}=\{10000,11100,11110,\\11101\}$. By the definition of the crossed cube, 11100 is adjacent to 11110.
Note that $|F_{1}|=4=5-1$. By Lemma \ref{2n-3}, $CQ_{5}^{1}-F_{1}$ is connected or has two components, one of which is an isolated vertex. Suppose that $CQ_{5}^{1}-F_{1}$ is disconnected. Let $u$ be the isolated vertex in $CQ_{5}^{1}-F_{1}$. Since $d_{CQ_{5}^{1}}(u)=4=|F_1|$, $u$ is connected to every vertex of $F_{1}$. So $u$ is adjacent to 11100 and 11110. Then we get that $CQ_{5}^{1}[\{u,11100,11110\}]$ is a triangle, a contradiction. Thus, $CQ_{5}^{1}-F_{1}$ is connected.
So we can conclude that $CQ_{n}^{1}-F_{1}$ is connected for $n\geq5$.
Note that $CQ_{n}^{0}-(A\cup F_{0})$ is connected.
Since $2^{n-1}-(4n-9)\geq1$ $(n\geq5)$, by Proposition \ref{match}, $CQ_{n}[V(CQ_{n}^{0}-(A\cup F_{0}))\cup V(CQ_{n}^{1}-F_{1})]=CQ_{n}-(A\cup F)$ is connected for $n\geq5$.
\end{proof}

%\begin{lemma}\label{a4n-9} The 3-extra connectivity $\tilde{\kappa}^{(3)}(CQ_{n})\leq4n-9$ for $n\geq4$.
%\end{lemma}

%\begin{proof} Let $A$ be defined in Lemma \ref{A} and $F=N_{CQ_{n}}(A)$. We have $|A|=4$, $|F|=4n-9$ and $CQ_{n}-A\cup F$ is connected. Note that $|V(CQ_{n}-A\cup F)|=2^{n}-(4+4n-9)\geq4$ for $n\geq4$. Thus, we get that $CQ_{n}-F$ has two components, one of which is $CQ_{n}[V(A)]$ and the other is $CQ_{n}-A\cup F$ with $|A|=4$ and
%$|V(CQ_{n}-A\cup F)|\geq4$.
%So $F$ is a 3-extra cut of $CQ_{n}$. By the definition of 3-extra connectivity, $\tilde{\kappa}^{(3)}(CQ_{n})\leq|F|=4n-9$.
%\end{proof}

\begin{lemma}(\cite{wm2})\label{n5f10} Let $CQ_{n}$ be the crossed cube and let $F\subseteq V(CQ_{5})$. If $|F|=10$, then $CQ_{5}-F$ satisfies one of the following conditions:
\\(1) $CQ_{5}-F$ is connected;
\\(2) $CQ_{5}-F$ has two components, one of which is a $K_{2}$;
\\(3) $CQ_{5}-F$ has two components, one of which is a 2-path;
\\(4) $CQ_{5}-F$ has two components, one of which is an isolated vertex;
\\(5) $CQ_{5}-F$ has three components, two of which are isolated vertices;
\\(6) $CQ_{5}-F$ has four components, three of which are isolated vertices;
\\(7) $CQ_{5}-F$ has three components, one of which is an isolated vertex and the other is a $K_{2}$.
\end{lemma}

\begin{lemma}\label{4n-10} Let $CQ_{n}$ be the crossed cube and let $F\subseteq V(CQ_{n})$ $(n\geq 5)$. If $3n-5\leq|F|\leq4n-10$, then $CQ_{n}-F$ satisfies one of the following conditions:
\\(1) $CQ_{n}-F$ is connected;
\\(2) $CQ_{n}-F$ has two components, one of which is a $K_{2}$;
\\(3) $CQ_{n}-F$ has two components, one of which is a 2-path;
\\(4) $CQ_{n}-F$ has two components, one of which is an isolated vertex;
\\(5) $CQ_{n}-F$ has three components, two of which are isolated vertices;
\\(6) $CQ_{n}-F$ has four components, three of which are isolated vertices;
\\(7) $CQ_{n}-F$ has three components, one of which is an isolated vertex and the other is a $K_{2}$.
\end{lemma}

\begin{proof} We prove the lemma by induction on $n$. By Lemma \ref{n5f10}, the lemma is true for $n=5$. We assume that the lemma is true for $n-1$, i.e., if $3n-8\leq|F|\leq4n-14$, then $CQ_{n-1}-F$ satisfies one of the conditions (1)-(7). Now we show that the lemma is also true for $n$ $(n\geq6)$. We can decompose $CQ_{n}$ along dimension $n-1$ into $CQ_{n}^{0}$ and $CQ_{n}^{1}$. Then both $CQ_{n}^{0}$ and $CQ_{n}^{1}$ are isomorphic to $CQ_{n-1}$. Let $F_{0}=F\cap V(CQ_{n}^{0})$ and $F_{1}=F\cap V(CQ_{n}^{1})$ with $|F_{0}|\leq |F_{1}|$. Let $B_{i}$ be the maximum component of $CQ_{n}^{i}-F_{i}$ (If $CQ_{n}^{i}-F_{i}$ is connected, then let $B_{i}=CQ_{n}^{i}-F_{i}$) for $i\in\{0,1\}$.
Since $3n-5\leq|F|\leq4n-10$, we have $0\leq|F_{0}|\leq\frac{4n-10}{2}=2n-5$ and $n\leq\lceil\frac{3n-5}{2}\rceil\leq|F_{1}|\leq4n-10$ $(n\geq6)$.
We consider the following cases.

{\it Case 1.} $n\leq|F_{1}|\leq2n-5$.

Note that $|F_{i}|\leq2n-5=2(n-1)-3$ for $i\in\{0,1\}$.
By Lemma \ref{2n-3}, $CQ_{n}^{i}-F_{i}$ is connected or has two components, one of which is an isolated vertex.
Since $2^{n-1}-(4n-10)-2\geq1$ $(n\geq6)$, by Proposition \ref{match}, $CQ_{n}[V(B_{0})\cup V(B_{1})]$ is connected. Thus, $CQ_{n}-F$ satisfies one of the conditions (1)-(7).

{\it Case 2.} $2n-4\leq|F_{1}|\leq3n-9$.

In this case, $|F_{0}|\leq4n-10-(2n-4)=2n-6<2n-5$. By Lemma \ref{2n-3}, $CQ_{n}^{0}-F_{0}$ is connected or has two components, one of which is an isolated vertex.
By Lemma \ref{3n-6}, $CQ_{n}^{1}-F_{1}$ satisfies one of the following conditions:
\\(a) $CQ_{n}^{1}-F_{1}$ has two components, one of which is a $K_{2}$;
\\(b) $CQ_{n}^{1}-F_{1}$ has two components, one of which is an isolated vertex;
\\(c) $CQ_{n}^{1}-F_{1}$ has three components, two of which are isolated vertices.

Since $2^{n-1}-(4n-10)-3\geq1$ $(n\geq6)$, by Proposition \ref{match}, $CQ_{n}[V(B_{0})\cup V(B_{1})]$ is connected. Thus, $CQ_{n}-F$ satisfies one of the conditions (1)-(7).

{\it Case 3.} $3n-8\leq|F_{1}|\leq4n-14$.

By the inductive hypothesis, $CQ_{n}^{1}-F_{1}$ satisfies one of the conditions (1)-(7).
In this case, $|F_{0}|\leq4n-10-(3n-8)=n-2$. By Lemma \ref{k}, $CQ_{n}^{0}-F_{0}$ is connected. Since $2^{n-1}-(4n-10)-3\geq1$ $(n\geq6)$, by Proposition \ref{match}, $CQ_{n}[V(B_{0})\cup V(B_{1})]$ is connected. Thus, $CQ_{n}-F$ satisfies one of the conditions (1)-(7).

{\it Case 4.} $4n-13\leq|F_{1}|\leq4n-10$.

In this case, $|F_{0}|\leq4n-10-(4n-13)=3$. By Lemma \ref{k}, $CQ_{n}^{0}-F_{0}$ is connected. Suppose that $CQ_{n}^{1}-F_{1}$ is connected. Since $2^{n-1}-(4n-10)\geq1$ $(n\geq6)$, by Proposition \ref{match}, $CQ_{n}[V(CQ_{n}^{0}-F_{0})\cup V(CQ_{n}^{1}-F_{1})]=CQ_{n}-F$ is connected. Then we suppose that $CQ_{n}^{1}-F_{1}$ is disconnected. Let the components of $CQ_{n}^{1}-F_{1}$ be $C_{1}$, $C_{2}$, \ldots, $C_{k}$ $(k\geq2)$. Note that $|F_{0}|\leq3$.
If every component $C_{i}$ of $CQ_{n}^{1}-F_{1}$ such that $|V(C_{i})|\geq4$ for $i\in\{1,\ldots,k\}$, then $CQ_{n}[V(CQ_{n}^{0}-F_{0})\cup V(C_{1})\cup\cdots\cup V(C_{k})]=CQ_{n}-F$ is connected. Suppose that there is a components $C_{i}$ such that $|V(C_{i})|\leq3$ for $i\in\{1,\ldots,k\}$. If $N_{CQ_{n}}(V(C_{i}))\cap V(CQ_{n}^{0})\subseteq F_{0}$, then $C_{i}$ is a component of $CQ_{n}-F$ with $|V(C_{i})|\leq3$. Combining $|F_{0}|\leq3$, we get that $CQ_{n}-F$ satisfies one of the conditions (1)-(7).
\end{proof}

%\begin{lemma}\label{b4n-9} The 3-extra connectivity $\tilde{\kappa}^{(3)}(CQ_{n})\geq4n-9$ for $n\geq5$.
%\end{lemma}
%\begin{proof}  Let $F$ be the minimum 3-extra cut of $CQ_{n}$. If $|F|\leq4n-10$, by Lemma \ref{4n-10}, then $F$ is not a 3-extra cut of $CQ_{n}$. Thus, we have $|F|\geq4n-9$. By the definition of 3-extra connectivity, $\tilde{\kappa}^{(3)}(CQ_{n})=|F|\geq4n-9$.
%\end{proof}

%Combining Lemmas \ref{a4n-9} and \ref{b4n-9}, we have the following theorem.

\begin{theorem}(\cite{Zhu})\label{4n-9} Let $CQ_{n}$ be the crossed cube. Then $\tilde{\kappa}^{(3)}(CQ_{n})=4n-9$ for $n\geq 5$.
\end{theorem}

%\begin{theorem}(\cite{Zhu})\label{4n-9} Let $CQ_{n}$ be the crossed cube. Then $\tilde{\kappa}^{(g)}(CQ_{n})=n(g + 1) -\frac{1}{2}
%g(g + 3)$ for $n\geq 4$ and $0\leq g\leq n-4$.
%\end{theorem}

\begin{lemma}(\cite{wm2})\label{n4f6} Let $F\subseteq V(CQ_{4})$. If $|F|=6$, then $CQ_{4}-F$ satisfies one of the following conditions:
\\(1) $CQ_{4}-F$ is connected;
\\(2) $CQ_{4}-F$ has two components, one of which is a $K_{2}$;
\\(3) $CQ_{4}-F$ has two components, one of which is an isolated vertex;
\\(4) $CQ_{4}-F$ has three components, two of which are isolated vertices;
\\(5) $CQ_{4}-F$ has two components, which are two components of order 5.
\end{lemma}

\begin{theorem}\label{cq4} Let $CQ_{4}$ be the crossed cube. Then the 3-extra connectivity of $CQ_{4}$ is not 7.
\end{theorem}

\begin{proof} Let $F$ be the minimum 3-extra cut of $CQ_{4}$. If $|F|=6$, by Lemma \ref{n4f6}, then $CQ_{4}-F$ satisfies one of the conditions (1)-(5) in Lemma \ref{n4f6}. If $CQ_{4}-F$ satisfies the condition (5), then $F$ is a 3-extra cut. By the definition of 3-extra connectivity, $\tilde{\kappa}^{(3)}(CQ_{4})\leq|F|=6$. Thus, the 3-extra connectivity of $CQ_{4}$ is not 7.
\end{proof}

We give an example such that the 3-extra connectivity of $CQ_{4}$ is not 7.

In $CQ_{4}$ (see Fig.1), let $F=\{0100,0111,0011,1000,1110,1011\}$. Then $CQ_{4}-F$  has two components $A$ and $B$, where $V(A)=\{0001,0000,0010,0110,\\1010\}$ and $V(B)=\{0101,1111,1001,1101,1100\}$.
It is easy to see that $|F|=6$ and $|V(A)|=|V(B)|=5$.
Thus, $F$ is a 3-extra cut of $CQ_{4}$. By the definition of 3-extra connectivity, $\tilde{\kappa}^{(3)}(CQ_{4})\leq |F|=6$. In other words, the 3-extra connectivity of $CQ_{4}$ is not 7.

\begin{lemma}(\cite{wm2})\label{n6f4n-9} Let $F\subseteq V(CQ_{n})$ $(n\geq 6)$. If $3n-4\leq|F|\leq4n-9$, then $CQ_{n}-F$ satisfies one of the following conditions:
\\(1) $CQ_{n}-F$ is connected;
\\(2) $CQ_{n}-F$ has two components, one of which is a $K_{2}$;
\\(3) $CQ_{n}-F$ has two components, one of which is a $K_{1,3}$;
\\(4) $CQ_{n}-F$ has two components, one of which is a 2-path;
\\(5) $CQ_{n}-F$ has two components, one of which is a 3-path;
\\(6) $CQ_{n}-F$ has two components, one of which is an isolated vertex;
\\(7) $CQ_{n}-F$ has three components, two of which are isolated vertices;
\\(8) $CQ_{n}-F$ has four components, three of which are isolated vertices;
\\(9) $CQ_{n}-F$ has three components, one of which is an isolated vertex and the other is a $K_{2}$;
\\(10) $CQ_{n}-F$ has three components, one of which is an isolated vertex and the other is a 2-path.
\end{lemma}

\begin{theorem}\label{tcq7} For $n\geq 7$, the crossed cube $CQ_{n}$ is tightly $(4n-9)$ super 3-extra connected.
\end{theorem}

\begin{proof} Let $F$ be a minimum 3-extra cut of $CQ_{n}$. By Theorem \ref{4n-9}, $|F|=4n-9$.
We can decompose $CQ_{n}$ along dimension $n-1$ into $CQ_{n}^{0}$ and $CQ_{n}^{1}$. Then both $CQ_{n}^{0}$ and $CQ_{n}^{1}$ are  isomorphic to $CQ_{n-1}$. Let $F_{0}=F\cap V(CQ_{n}^{0})$ and $F_{1}=F\cap V(CQ_{n}^{1})$ with $|F_{0}|\leq |F_{1}|$. Then $|F_{1}|\geq\lceil\frac{4n-9}{2}\rceil=2n-4$.
We consider the following cases.

{\it Case 1.}  $2n-4\leq|F_{1}|\leq3n-9$.

In this case, $|F_{0}|\leq4n-9-(2n-4)=2n-5$. By Lemma \ref{2n-3}, $CQ_{n}^{0}-F_{0}$ is connected or has two components, one of which is an isolated vertex.
By Lemma \ref{3n-6}, $CQ_{n}^{1}-F_{1}$ satisfies one of the following conditions:
\\(1) $CQ_{n}^{1}-F_{1}$ is connected;
\\(2) $CQ_{n}^{1}-F_{1}$ has two components, one of which is a $K_{2}$;
\\(3) $CQ_{n}^{1}-F_{1}$ has two components, one of which is an isolated vertex;
\\(4) $CQ_{n}^{1}-F_{1}$ has three components, two of which are isolated vertices.

Let $B_{i}$ be the maximum component of $CQ_{n}^{i}-F_{i}$ for $i\in\{0,1\}$ (if $CQ_{n}^{i}-F_{i}$ is connected, then $B_{i}=CQ_{n}^{i}-F_{i}$).
Since $2^{n-1}-(4n-9)-3\geq1$ $(n\geq7)$, by Proposition \ref{match}, $CQ_{n}[V(B_{0})\cup V(B_{1})]$ is connected. Thus, $F$ is not a 3-extra cut of $CQ_{n}$. This is a contradiction to that $F$ is a minimum 3-extra cut of $CQ_{n}$.

{\it Case 2.} $3n-8\leq|F_{1}|\leq4n-14$.

In this case, $|F_{0}|=4n-9-(3n-8)=n-1$. By Lemma \ref{2n-3}, $CQ_{n}^{0}-F_{0}$ is connected or has two components, one of which is an isolated vertex. Note that $3(n-1)-5=3n-8\leq|F_{1}|\leq4n-14=4(n-1)-10$.
By Lemma \ref{4n-10}, $CQ_{n}^{1}-F_{1}$ satisfies one of the following conditions:
\\(1) $CQ_{n}^{1}-F_{1}$ is connected;
\\(2) $CQ_{n}^{1}-F_{1}$ has two components, one of which is a $K_{2}$;
\\(3) $CQ_{n}^{1}-F_{1}$ has two components, one of which is a 2-path;
\\(4) $CQ_{n}^{1}-F_{1}$ has two components, one of which is an isolated vertex;
\\(5) $CQ_{n}^{1}-F_{1}$ has three components, two of which are isolated vertices;
\\(6) $CQ_{n}^{1}-F_{1}$ has four components, three of which are isolated vertices;
\\(7) $CQ_{n}^{1}-F_{1}$ has three components, one of which is an isolated vertex and the other is a $K_{2}$.

If $CQ_{n}^{0}-F_{0}$ is connected, by Proposition \ref{match}, then $F$ is not a 3-extra cut of $CQ_{n}$.
We suppose that $CQ_{n}^{0}-F_{0}$ is disconnected. Let $u$ be the isolated vertex and $B_{0}$ be the other component in $CQ_{n}^{0}-F_{0}$. If $CQ_{n}^{1}-F_{1}$ satisfies the condition (3), then let $P$ be the 2-path and $B_{1}$ be the other component in $CQ_{n}^{1}-F_{1}$. Since $2^{n-1}-(4n-9)-4\geq1$ $(n\geq7)$, by Proposition \ref{match}, $CQ_{n}[V(B_{0})\cup V(B_{1})]$ is connected. If $u$ is connected to $P$, then $CQ_{n}-F$ has two components, one of which is $CQ_{n}[V(B_{0})\cup V(B_{1})]$ and the other is $CQ_{n}[\{u\}\cup V(P)]$ with $|\{u\}\cup V(P)|=1+3=4$. If $CQ_{n}^{1}-F_{1}$ satisfies one of the conditions (1)-(7) except (3), then $F$ is not a minimum 3-extra cut of $CQ_{n}$, a contradiction.

{\it Case 3.} $|F_{1}|=4n-13$.

In this case, $|F_{0}|=4n-9-(4n-13)=4$. By Lemma \ref{k}, $CQ_{n}^{0}-F_{0}$ is connected.
Note that $|F_{1}|=4n-13=4(n-1)-9$.
By Lemma \ref{n6f4n-9}, $CQ_{n}^{1}-F_{1}$ satisfies one of the following conditions:
\\(1) $CQ_{n}^{1}-F_{1}$ is connected;
\\(2) $CQ_{n}^{1}-F_{1}$ has two components, one of which is a $K_{2}$;
\\(3) $CQ_{n}^{1}-F_{1}$ has two components, one of which is a $K_{1,3}$;
\\(4) $CQ_{n}^{1}-F_{1}$ has two components, one of which is a 2-path;
\\(5) $CQ_{n}^{1}-F_{1}$ has two components, one of which is a 3-path;
\\(6) $CQ_{n}^{1}-F_{1}$ has two components, one of which is an isolated vertex;
\\(7) $CQ_{n}^{1}-F_{1}$ has three components, two of which are isolated vertices;
\\(8) $CQ_{n}^{1}-F_{1}$ has four components, three of which are isolated vertices;
\\(9) $CQ_{n}^{1}-F_{1}$ has three components, one of which is an isolated vertex and the other is a $K_{2}$;
\\(10) $CQ_{n}^{1}-F_{1}$ has three components, one of which is an isolated vertex and the other is a 2-path.

Suppose that $CQ_{n}^{1}-F_{1}$ satisfies the condition (3), then $CQ_{n}^{1}-F_{1}$ has two components, one of which is a $K_{1,3}$. Note that $|V(K_{1,3})|=4$. If $N_{CQ_{n}}(V(K_{1,3}))\cap V(CQ_{n}^{0})=F_{0}$, then $CQ_{n}-F$ has two components, one of which is a $K_{1,3}$. Similarly, if $CQ_{n}^{1}-F_{1}$ satisfies the condition (5),
then $CQ_{n}-F$ has two components, one of which is a subgraph $H$ of $CQ_{n}$ with $|V(H)|=4$.
If $CQ_{n}^{1}-F_{1}$ satisfies one of the conditions (1)-(10) except (3) and (5), then $F$ is not a minimum 3-extra cut of $CQ_{n}$, a contradiction.

{\it Case 4.} $4n-12\leq|F_{1}|\leq4n-9$.

In this case, $|F_{0}|\leq4n-9-(4n-12)=3$. By Lemma \ref{k}, $CQ_{n}^{0}-F_{0}$ is connected.
Suppose that $CQ_{n}^{1}-F_{1}$ is connected. Since $2^{n-1}-(4n-9)\geq1$ $(n\geq7)$,
by Proposition \ref{match}, $CQ_{n}[V(CQ_{n}^{0}-F_{0})\cup V(CQ_{n}^{1}-F_{1})]=CQ_{n}-F$ is connected. Then we suppose that $CQ_{n}^{1}-F_{1}$ is disconnected.
Let the components of $CQ_{n}^{1}-F_{1}$ be $C_{1}$, $C_{2}$, \ldots, $C_{k}$ $(k\geq2)$.
Note that $|F_{0}|\leq3$.
If every component $C_{i}$ of $CQ_{n}^{1}-F_{1}$ such that $|V(C_{i})|\geq4$ for $i\in\{1,\ldots,k\}$, then $CQ_{n}[V(CQ_{n}^{0}-F_{0})\cup V(C_{1})\cup\cdots\cup V(C_{k})]=CQ_{n}-F$ is connected. Suppose that there is a components $C_{i}$ such that $|V(C_{i})|\leq3$ for $i\in\{1,\ldots,k\}$. If $N_{CQ_{n}}(V(C_{i}))\cap V(CQ_{n}^{0})\subseteq F_{0}$, then $C_{i}$ is a component of $CQ_{n}-F$ with $|V(C_{i})|\leq3$. Thus, $F$ is not a 3-extra cut of $CQ_{n}$. This is a contradiction to that $F$ is a minimum 3-extra cut of $CQ_{n}$.
\end{proof}

\section{The 3-extra diagnosaility of crossed cubes under the PMC and MM$^*$ model}

\begin{theorem}(\cite{da,yu,za})\label{pmc} A system $G=(V,E)$ is $g$-extra $t$-diagnosable under the PMC model if and only if there is an edge $uv\in E$ with $u\in V\setminus(F_{1}\cup F_{2})$ and $v\in F_{1}\triangle F_{2}$ for each distinct pair of $g$-extra faulty subsets $F_{1}$ and $F_{2}$ of $V(CQ_{n})$ with $|F_{1}|\leq t$ and $|F_{2}|\leq t$ (see Fig.3).
\end{theorem}

\vspace{1cm}

%TeXCAD Picture [pmc.pic]. Options:
%\grade{\on}
%\emlines{\off}
%\epic{\off}
%\beziermacro{\on}
%\reduce{\on}
%\snapping{\off}
%\quality{8.00}
%\graddiff{0.01}
%\snapasp{1}
%\zoom{4.0000}
\unitlength 1mm % = 2.85pt
\linethickness{0.4pt}
\ifx\plotpoint\undefined\newsavebox{\plotpoint}\fi % GNUPLOT compatibility
\begin{picture}(61.5,32.46)(0,0)
\put(43.5,25.46){\oval(16,8)[]}
\put(53.5,25.46){\oval(16,8)[]}
\put(40,32.46){$F_{1}$}
\put(53,32.46){$F_{2}$}
\put(13.75,25.2){\oval(16,8)[]}
\put(23.75,25.2){\oval(16,8)[]}
\put(10.75,24.2){\line(0,-1){8}}
\put(10.75,24.2){\circle*{1}}
\put(10.75,16.2){\circle*{1}}
\put(8.25,24.95){$v$}
\put(8.25,16.7){$u$}
\put(10,32.2){$F_{1}$}
\put(23,32.2){$F_{2}$}
\put(56,24.46){\line(0,-1){8}}
\put(56,24.46){\circle*{1}}
\put(56,16.46){\circle*{1}}
\put(58.75,16.25){\makebox(0,0)[cc]{$u$}}
\put(58.25,24.5){\makebox(0,0)[cc]{$v$}}
\put(40,6){\makebox(0,0)[cc]{Fig.3. $u$ diagnosis $v$ under the PMC model.}}
\put(33.6,25){\makebox(0,0)[cc]{or}}
\end{picture}

\begin{theorem}(\cite{da,yu,za}) \label{mmt} A system $G = (V, E)$ is $g$-extra $t$-diagnosable
under the MM$^*$ model if and only if each distinct pair of $g$-extra faulty subsets $F_1$ and $F_2$ of $V$ with $|F_1|\leq t$ and $|F_2| \leq t$ satisfies one of the following conditions (see Fig.4):

(1) There exist two vertices $u,w\in V(G)\setminus (F_{1}\cup F_{2})$  and there exists a vertex $v\in F_{1}\triangle F_{2}$ such that $uw, vw\in E(G)$.

(2) There exist two vertices $u,v\in F_{1}\setminus F_{2}$ and there exists a vertex $w\in V(G)\setminus (F_{1}\cup F_{2})$ such that $uw,vw\in E(G)$.

(3) There exist two vertices $u,v\in F_{2}\setminus F_{1}$ and there exists a vertex $w\in V(G)\setminus (F_{1}\cup F_{2})$ such that $uw,vw\in E(G)$.
\end{theorem}

%TeXCAD Picture [MMMM.pic]. Options:
%\grade{\on}
%\emlines{\off}
%\epic{\off}
%\beziermacro{\on}
%\reduce{\on}
%\snapping{\off}
%\quality{8.00}
%\graddiff{0.01}
%\snapasp{1}
%\zoom{4.0000}
\unitlength 1mm % = 2.85pt
\linethickness{0.4pt}
\ifx\plotpoint\undefined\newsavebox{\plotpoint}\fi % GNUPLOT compatibility
\begin{picture}(49.5,38.25)(0,0)
\put(22,28.13){\oval(24.5,9.75)[]}
\put(37,28){\oval(25,10)[]}
%\emline(18,28.5)(16.25,17.75)
\multiput(18,28.5)(-.0336538,-.2067308){52}{\line(0,-1){.2067308}}
%\end
%\emline(16.25,17.75)(14,28.5)
\multiput(16.25,17.75)(-.0335821,.1604478){67}{\line(0,1){.1604478}}
%\end
\put(36.25,28.5){\line(0,-1){10.75}}
%\emline(36.25,18)(32,13)
\multiput(36.25,18)(-.03373016,-.03968254){126}{\line(0,-1){.03968254}}
%\end
%\emline(40.75,28.25)(41.75,17.75)
\multiput(40.75,28.25)(.033333,-.35){30}{\line(0,-1){.35}}
%\end
%\emline(41.75,17.75)(44.5,28.25)
\multiput(41.75,17.75)(.0335366,.1280488){82}{\line(0,1){.1280488}}
%\end
\put(22.75,28.5){\line(0,-1){10.25}}
\put(22.75,18.25){\line(3,-4){3.75}}
\put(14,27.75){\circle*{1}}
\put(18.25,28.25){\circle*{1}}
\put(16.25,18){\circle*{1}}
\put(22.5,28.5){\circle*{1}}
\put(22.75,18){\circle*{1}}
\put(26.25,13.5){\circle*{1}}
\put(32.5,13.25){\circle*{1}}
\put(36.25,17.75){\circle*{1}}
\put(36.25,28.5){\circle*{1}}
\put(40.75,28){\circle*{1}}
\put(44.25,27.75){\circle*{1}}
\put(42,18){\circle*{1}}
\put(19.5,38){\makebox(0,0)[ct]{$F_{1}$}}
\put(39.25,38.25){\makebox(0,0)[ct]{$F_{2}$}}
\put(14,29.25){\makebox(0,0)[cb]{$u$}}
\put(18.25,29.25){\makebox(0,0)[cb]{$v$}}
\put(16,17.25){\makebox(0,0)[ct]{$w$}}
\put(22.5,29.25){\makebox(0,0)[cb]{$v$}}
\put(22,16.75){\makebox(0,0)[rc]{$w$}}
\put(36.25,29.5){\makebox(0,0)[cb]{$v$}}
\put(40.75,29.25){\makebox(0,0)[cb]{$v$}}
\put(44.25,29.5){\makebox(0,0)[cb]{$u$}}
\put(36.75,16.75){\makebox(0,0)[lc]{$w$}}
\put(33.5,13){\makebox(0,0)[lc]{$u$}}
\put(42,16.6){\makebox(0,0)[cc]{$w$}}
\put(25,13.5){\makebox(0,0)[rc]{$u$}}
\put(24,22){\makebox(0,0)[lt]{(1)}}
\put(12.75,22){\makebox(0,0)[ct]{(2)}}
\put(34.75,22){\makebox(0,0)[rt]{(1)}}
\put(43.5,22){\makebox(0,0)[lt]{(3)}}
\put(0,7){Fig.4. $w$ diagnosis $u$ and $v$}
\put(10,2){under the MM* model.}
\end{picture}

\begin{lemma}\label{pmct1} Let $n\geq4$. Then the 3-extra diagnosability of the crossed cube $CQ_{n}$ under the PMC and MM* model is less than or equal to $4n-6$, i.e.,  $\tilde{t_{3}}(CQ_{n})\leq 4n-6$.
\end{lemma}

\begin{proof} Let $A$ be defined in Lemma \ref{A}, $F_{1}=N_{CQ_{n}}(A)$ and $F_{2}=A\cup F_{1}$. By Lemma \ref{A}, $|F_{1}|=4n-9$, $|F_{2}|=4n-5$ and $CQ_{n}-F_{2}$ is connected. Thus,
$CQ_{n}-F_{1}$ has two components $CQ_{n}-F_{2}$ and $CQ_{n}[A]$. Note that $|A|=4$ and $|V(CQ_{n}-F_{2})|=2^{n}-(4n-5)\geq4$ for $n\geq4$. By the definition of 3-extra connectivity,
$F_{1}$ is a 3-extra cut of $CQ_{n}$.
Thus, $F_{1}$ and $F_{2}$ are both 3-extra faulty sets of $CQ_{n}$ with $|F_{1}|=4n-9$ and $|F_{2}|=4n-5$. Since $A=F_{1}\triangle F_{2}$ and $N_{CQ_{n}}(A)=F_{1}\subset F_{2}$, there is no edge of $CQ_{n}$ between $V(CQ_{n})\setminus (F_{1}\cup F_{2})$ and $F_{1}\triangle F_{2}$. By Theorems \ref{pmc} and \ref{mmt},  $CQ_{n}$ is not 3-extra $(4n-5)$-diagnosable under PMC and MM* model, respectively. By the definition of 3-extra diagnosability, we can deduce that the 3-extra diagnosability of $CQ_{n}$ is less than or equal to $4n-6$, i.e., $\tilde{t_{3}}(CQ_{n})\leq 4n-6$.
\end{proof}

\begin{lemma}\label{pmct2} Let $n\geq5$. Then the 3-extra diagnosability of the crossed cube $CQ_{n}$ under the PMC model is more than or equal to $4n-6$, i.e., $\tilde{t_{3}}(CQ_{n})\geq 4n-6$.
\end{lemma}

\begin{proof} By the definition of the 3-extra diagnosability, it is sufficient to show that $CQ_{n}$ is 3-extra $(4n-6)$-diagnosable. By Theorem \ref{pmc}, we need to prove that there is an edge $uv\in E$ with $u\in V(CQ_{n})\setminus(F_{1}\cup F_{2})$ and $v\in F_{1}\triangle F_{2}$ for each distinct pair of 3-extra faulty subsets $F_{1}$ and $F_{2}$ of $V(CQ_{n})$ with $|F_{1}|\leq 4n-6$ and $|F_{2}|\leq 4n-6$.

Suppose, on the contrary,  that there are two distinct 3-extra faulty subsets $F_{1}$ and $F_{2}$ of $V(CQ_{n})$ with $|F_{1}|\leq 4n-6$ and $|F_{2}|\leq 4n-6$, but there is no edge between $V(CQ_{n})\setminus(F_{1}\cup F_{2})$ and $F_{1}\triangle F_{2}$. Without loss of generality, assume that $F_{2}\setminus F_{1}\neq \emptyset$.

{\it Claim 1.} $V(CQ_{n})\neq F_{1}\cup F_{2}$.

On the contrary, we suppose that $V(CQ_{n})= F_{1}\cup F_{2}$. We get
$2^{n}=|V(CQ_{n})|=|F_{1}\cup F_{2}|=|F_{1}|+|F_{2}|-|F_{1}\cap F_{2}|\leq |F_{1}|+|F_{2}|\leq 2(4n-6)=8n-12$, a contradiction to $n\geq5$. Therefore, $V(CQ_{n})\neq F_{1}\cup F_{2}$. The proof of Claim 1 is complete.

Since there is no edge between $V(CQ_{n})\setminus(F_{1}\cup F_{2})$ and $F_{1}\triangle F_{2}$, $CQ_{n}-F_{1}$ has two parts $CQ_{n}\setminus (F_{1}\cup F_{2})$ and $CQ_{n}[F_{2}\setminus F_{1}]$. Note that $F_{1}$ is a 3-extra faulty set. Thus, every component $B_{i}$ of $CQ_{n}\setminus (F_{1}\cup F_{2})$ such that $|V(B_{i})|\geq4$ and every component $C_{i}$ of $CQ_{n}[F_{2}\setminus F_{1}]$ such that $|V(C_{i})|\geq4$.
If $F_{1}\setminus F_{2}=\emptyset$, then $F_{1}\cap F_{2}=F_{1}$. Thus, $F_{1}\cap F_{2}$ is also a 3-extra faulty set. If $F_{1}\setminus F_{2}\neq\emptyset$, similarly, every component $D_{i}$ of $CQ_{n}[F_{1}\setminus F_{2}]$ such that $|V(D_{i})|\geq4$. Thus, $F_{1}\cap F_{2}$ is also a 3-extra faulty set. Since there is no edge between $V(CQ_{n})\setminus(F_{1}\cup F_{2})$ and $F_{1}\triangle F_{2}$, $F_{1}\cap F_{2}$ is a 3-extra cut of $CQ_{n}$. By Theorem \ref{4n-9}, $|F_{1}\cap F_{2}|\geq 4n-9$. Thus, $|F_{2}|=|F_{2}\setminus F_{1}|+|F_{1}\cap F_{2}|\geq 4+4n-9=4n-5$. This is a contradiction to that $|F_{2}|\leq 4n-6$. Therefore, $CQ_{n}$ is 3-extra $(4n-6)$-diagnosable, i.e., $\tilde{t_{3}}(CQ_{n})\geq 4n-6$.
\end{proof}

Combining Lemmas \ref{pmct1} and \ref{pmct2}, we have the following theorem.

\begin{theorem} Let $n\geq5$. Then the 3-extra diagnosability of the crossed cube $CQ_{n}$ under the PMC model is $4n-6$, i.e., $\tilde{t_{3}}(CQ_{n})= 4n-6$.
\end{theorem}

A component of a graph $G$ is odd according as it has an odd number of vertices. We denote by $o(G)$ the number of odd components of $G$.

\begin{lemma}(\cite{bj})\label{odd} A graph $G=(V,E)$ has a perfect matching if and only if $o(G-S)\leq |S|$ for all $S\subseteq V$.
\end{lemma}

\begin{lemma}\label{mm2} Let $n\geq7$. Then the 3-extra diagnosability of the crossed cube $CQ_{n}$ under the MM* model is more than or equal to $4n-6$, i.e., $\tilde{t_{3}}(CQ_{n})\geq 4n-6$.
\end{lemma}

\begin{proof} By the definition of 3-extra diagnosability, it is sufficient to show that $CQ_{n}$ is 3-extra $(4n-6)$-diagnosable. On the contrary, there are two distinct 3-extra faulty subsets $F_{1}$ and $F_{2}$ of $CQ_{n}$ with $|F_{1}|\leq 4n-6$ and $|F_{2}|\leq 4n-6$, but the vertex set pair $(F_{1},F_{2})$ is not satisfied with any one condition in Theorem \ref{mmt}. Without loss of generality, assume that $F_{2}\setminus F_{1}\neq\emptyset$.
By Claim 1 in Lemma 4.2, we get $V(CQ_{n})\neq F_{1}\cup F_{2}$.

{\it Claim 1.} $CQ_{n}-(F_{1}\cup F_{2})$ has no isolated vertex.

On the contrary, we suppose that $CQ_{n}-(F_{1}\cup F_{2})$ has at least one isolated vertex $w$. Since $F_{1}$ is one 3-extra faulty set, there is a vertex $u\in F_{2}\setminus F_{1}$ such that $u$ is adjacent to $w$. Note that the vertex set pair $(F_{1},F_{2})$ is not satisfied with any one condition in Theorem \ref{mmt}. By the condition (3) of Theorem \ref{mmt}, there is at most one vertex $u\in F_{2}\setminus F_{1}$ such that $u$ is adjacent to $w$. Thus, there is just a vertex $u\in F_{2}\setminus F_{1}$ such that $u$ is adjacent to $w$.
If $F_{1}\setminus F_{2}=\emptyset$, then $F_{1}\subseteq F_{2}$. Since $F_{2}$ is a 3-extra faulty set, every component $G_{i}$ of $CQ_{n}-F_{2}$ has $|V(G_{i})|\geq4$. Note that $CQ_{n}-F_{2}=CQ_{n}-(F_{1}\cup F_{2})$. So $CQ_{n}-(F_{1}\cup F_{2})$ has no isolated vertex. It is contradict with the hypothesis. Thus, $F_{1}\setminus F_{2}\neq\emptyset$.
Similarly, we can deduce that there is just a vertex $v\in F_{1}\setminus F_{2}$ such that $v$ is adjacent to $w$. Let $W\subseteq V(CQ_{n})\setminus (F_{1}\cup F_{2})$ be the set of isolated vertices in $CQ_{n}[V(CQ_{n})\setminus (F_{1}\cup F_{2})]$, and let $H$ be the induced subgraph by the vertex set $V(CQ_{n})\setminus (F_{1}\cup F_{2}\cup W)$. Then for any vertex $w\in W$, we can get that $w$ has $(n-2)$ neighbors in $F_{1}\cap F_{2}$. By Lemma \ref{odd} and Proposition \ref{match}, $|W|\leq o(CQ_{n}-(F_{1}\cup F_{2}))\leq |F_{1}\cup F_{2}|=|F_{1}|+|F_{2}|-|F_{1}\cap F_{2}|\leq 2(4n-6)-(n-2)=7n-10$. We assume that $V(H)=\emptyset$. Then $2^{n}=|V(CQ_{n})|=|F_{1}\cup F_{2}|+|W|=|F_{1}|+|F_{2}|-|F_{1}\cap F_{2}|+|W|\leq 2(4n-6)-(n-2)+(7n-10)=14n-20$, a contradiction to that $n\geq7$. Therefore, $V(H)\neq\emptyset$.

Since the vertex set pair $(F_{1},F_{2})$ is not satisfied with the condition (1) of Theorem \ref{mmt} and any vertex of $V(H)$ is not isolated in $H$, we induce that there is no edge between $V(H)$ and $F_{1}\triangle F_{2}$. If $F_{1}\cap F_{2}=\emptyset$, then $F_{1}\triangle F_{2}=F_{1}\cup F_{2}$. Since there is no edge between $V(H)$ and $F_{1}\triangle F_{2}$,
$CQ_{n}$ is disconnected, a contradiction to that $CQ_{n}$ is connected. Thus, $F_{1}\cap F_{2}\neq\emptyset$.
Note that $CQ_{n}-(F_{1}\cap F_{2})$ has two parts: $H$ and $CQ_{n}[(F_{1}\setminus F_{2})\cup (F_{2}\setminus F_{1})\cup W]$. Thus, $F_{1}\cap F_{2}$ is a vertex cut of $CQ_{n}$.
Since $F_{1}$ is a 3-extra faulty set of $CQ_{n}$, we have that every component $H_{i}$ of $H$ has $|V(H_{i})|\geq4$ and every component $B_{i}^{2}$ of $CQ_{n}[(F_{2}\setminus F_{1})\cup W]$ has $|V(B_{i}^{2})|\geq4$. Similarly, every component $B_{i}^{1}$ of $CQ_{n}[(F_{1}\setminus F_{2})\cup W]$ has $|V(B_{i}^{1})|\geq4$. By the definition of 3-extra cut, $F_{1}\cap F_{2}$ is a 3-extra cut of $CQ_{n}$. By Theorem \ref{4n-9}, $|F_{1}\cap F_{2}|\geq 4n-9$.

Note that any vertex $w\in W$ has two neighbors $u$ and $v$ such that $u\in V(F_{1}\setminus F_{2})$ and $v\in V(F_{2}\setminus F_{1})$. If there is not a vertex $w$ such that $w\in V(B_{i}^{2})$, then $B_{i}^{2}$ is a component of $F_{2}\setminus F_{1}$. We get that
$|F_{2}\setminus F_{1}|\geq|V(B_{i}^{2})|\geq4$.
If there is a vertex $w$ such that $w\in V(B_{i}^{2})$, then $B_{i}^{2}-w$ is a component of $F_{2}\setminus F_{1}$. We get that $|F_{2}\setminus F_{1}|\geq |V(B_{i}^{2})\setminus\{u\}|\geq3$. Thus, $|F_{2}\setminus F_{1}|\geq3$. Similarly, $|F_{1}\setminus F_{2}|\geq3$.
Since $|F_{1}\cap F_{2}|=|F_{2}|-|F_{2}\setminus F_{1}|\leq (4n-6)-3=4n-9$, we have $|F_{1}\cap F_{2}|=4n-9$. Then  we get that $|F_{2}\setminus F_{1}|=3$ and $|F_{2}|=4n-6$. Similarly, we have $|F_{1}\setminus F_{2}|=3$ and $|F_{1}|=4n-6$. By Theorem \ref{tcq7}, the crossed cube $CQ_{n}$ is tightly $(4n-9)$ super 3-extra connected, i.e., $CQ_{n}-(F_{1}\cap F_{2})$ has two components, one of which is a subgraph of order 4. Thus,  $CQ_{n}-(F_{1}\cap F_{2})$ has two components, one of which is $CQ_{n}[(F_{1}\setminus F_{2})\cup (F_{2}\setminus F_{1})\cup W]$ and the other one is $H$ with $|V(H)|=4$ and $|(F_{1}\setminus F_{2})\cup (F_{2}\setminus F_{1})\cup W|\geq 3+3+1=7$. Note that $|W|\leq7n-10$.
$2^{n}=|V(CQ_{n})|=|F_{1}\setminus F_{2}|+|F_{2}\setminus F_{1}|+|F_{1}\cap F_{2}|+|W|+|V(H)|\leq3+3+(4n-9)+(7n-10)+4=11n-9$, a contradiction to $n\geq6$. The proof of Claim 1 is complete.

Let $u\in V(CQ_{n})\setminus (F_{1}\cup F_{2})$. By Claim 1 in Lemma \ref{mm2}, $u$ has at least one neighbor in $CQ_{n}-(F_{1}\cup F_{2})$. Since $(F_{1},F_{2})$ is not satisfied with any one condition in Theorem \ref{mmt}, $u$ has no neighbor in $F_{1}\vartriangle F_{2}$. By the arbitrariness of $u$, there is no edge between $V(CQ_{n})\setminus (F_{1}\cup F_{2})$ and $F_{1}\vartriangle F_{2}$.
Since $F_{1}$ and $F_{2}$ are two 3-extra faulty set,
every component $H_{i}$ of $CQ_{n}-(F_{1}\cup F_{2})$ has $|V(H_{i})|\geq4$, every component $B_{i}$ of $CQ_{n}([F_{2}\setminus F_{1}])$ has $|V(B_{i})|\geq4$, and every component $C_{i}$ of $CQ_{n}([F_{1}\setminus F_{2}])$ has $|V(C_{i})|\geq4$ when $F_{1}\setminus F_{2}\neq\emptyset$. Thus, $F_{1}\cap F_{2}$ is also a 3-extra faulty set. Since there is no edge between $V(CQ_{n}\setminus(F_{1}\cup F_{2}))$ and $F_{1}\vartriangle F_{2}$, we have $F_{1}\cap F_{2}$ is a 3-extra cut of $CQ_{n}$.
By Theorem \ref{4n-9}, we have $|F_{1}\cap F_{2}|\geq 4n-9$. Since $B_{i}$ is a component of $CQ_{n}([F_{2}\setminus F_{1}])$ with $|V(B_{i})|\geq4$, we have $|F_{2}\setminus F_{1}|\geq |V(B_{i})|\geq4$.
Therefore, $|F_{2}|=|F_{2}\setminus F_{1}|+|F_{1}\cap F_{2}|\geq 4+(4n-9)=4n-5$, which contradicts $|F_{2}|\leq 4n-6$. Thus, $CQ_{n}$ is 3-extra $(4n-6)$-diagnosable, i.e., $\tilde{t_{3}}(CQ_{n})\geq 4n-6$. The proof is complete.
\end{proof}

Combining Lemmas \ref{pmct1} and \ref{mm2}, we can get the following theorem.

\begin{theorem} Let $n\geq7$. Then the 3-extra diagnosability of the crossed cube $CQ_{n}$ under the MM* model is $4n-6$, i.e., $\tilde{t_{2}}(CQ_{n})=4n-6$.
\end{theorem}

\section{Conclusions}

We prove that the 3-extra connectivity of $CQ_{n}$ is $4n-9$ for $n\geq5$. Moreover, $CQ_{n}$ is tightly $(4n-9)$ super 3-extra connected for $n\geq 7$. Then we determine that the 3-extra diagnosability of $CQ_{n}$ is $4n-6$ under the PMC model $(n\geq5)$ and MM* model $(n\geq7)$. On the basis of this study, the researchers can continue to study the $g$-extra connectivity and diagnosability of networks.

\section*{Acknowledgements}

The authors would like to thank the Natural Science Foundation of China
(61370001).


\begin{thebibliography}{00}



\bibitem[{bj}{bj}(2007)]{bj} J.A. Bondy, U.S.R. Murty, Graph Theory, Springer, New York, 2007.
\bibitem[{ch}{ch}(2013)]{ch} N.W. Chang, S.Y. Hsieh, $\{2,3\}$-Extraconnectivities of hypercube-like
networks, Journal of Computer and System Sciences 79 (2013) 669-688.


\bibitem[{ct}{ct}(2013)]{ct} N.W. Chang, C.Y. Tsai, S.Y. Hsieh, On 3-extra connectivity and 3-extra edge connectivity of folded hypercubes, IEEE Transactions on Computers 63 (6) (2013) 1594-1600.



\bibitem[{da}{da}(1984)]{da} A.T. Dahbura, G.M. Masson, An $O(n^{2.5})$ fault identification algorithm for diagnosable systems, IEEE Transactions on Computers 33 (6) (1984) 486-492.



\bibitem[{ef}{ef}(1992)]{ef} K. Efe, Member, The crossed cube architecture for parallel computation, IEEE Transactions on Parallel and Distributed Systems 3 (5) (1992) 513-524.


\bibitem[{jf}{jf}(1996)]{jf} J. F$\grave{a}$brega, M.A. Fiol, On the extraconnectivity of graphs, Discrete Mathematics 155 (1996) 49-57.

\bibitem[{gh1}{gh1}(2014)]{gh1} M.M. Gu, R.X. Hao, 3-extra connectivity of 3-ary $n$-cube networks, Information Processing Letters 114 (2014) 486-491.


\bibitem[{gh2}{gh2}(2017)]{gh2} M.M. Gu, R.X. Hao, J.B. Liu, On the extraconnectivity of $k$-ary $n$-cube networks, International Journal of Computer Mathematics 94 (1) (2017) 95-106.

\bibitem[{lm}{lm}(2011)]{lm} H.Z. Li, J.X. Meng, W.H. Yang, 3-extra connectivity of Cayley graphs
generated by transposition generating trees, Journal of Xinjiang University 28 (2) (2011) 149-151.


\bibitem[{mm}{mm}(1981)]{mm} J. Maeng, M. Malek, A comparison connection assignment for self-diagnosis of multiprocessors systems, in: Proceeding of the 11th International Symposium on Fault-Tolerant Computing, pp. 173-175, 1981.


\bibitem[{pp}{pp}(1967)]{pp} F.P. Preparata, G. Metze, and R.T. Chien, On the connection assignment problem of diagnosable systems, IEEE Transactions on Computers,  EC-16 (6) (1967) 848-854.

\bibitem[{Ren2}{Ren2}(2017)]{Ren2} Yunxia Ren, Shiying Wang, The tightly super 2-extra connectivity and 2-extra diagnosability of locally twisted cubes, Journal of Interconnection Networks 17 (2) (2017) 1750006.

\bibitem[{sd}{sd}(1992)]{sd} A. Sengupta, A.T. Dahbura, On self-diagnosable multiprocessor systems: diagnosis by the comparison approach, IEEE Transactions on Computers 41 (11) (1992) 1387-1396.


\bibitem[{sWang}{sWang}(2016)]{sWang} S.Y. Wang, Z.H. Wang, M.J.S. Wang, The 2-extra connectivity and 2-extra diagnosability of bubble-sort star graph networks, The Computer Journal 59 (12) (2016) 1839-1856.

\bibitem[{Wang4}{Wang4}(2017)]{Wang4} S.Y. Wang and Y.X. Yang,  The 2-good-neighbor (2-extra) diagnosability of alternating group graph networks under the PMC model and MM* model, Applied Mathematics and Computation 305 (2017) 241-250.

\bibitem[{wm1}{wm1}(2017)]{wm1} S.Y. Wang, X.L. Ma, The tightly super 2-extra connectivity and 2-extra diagnosability of crossed cubes (to appear).

\bibitem[{wm2}{wm2}(2017)]{wm2} S.Y. Wang, X.L. Ma, Y.X. Ren, The tightly super 2-good-neighbor connectivity and 2-good-neighbor diagnosability of crossed cubes, International Journal of New Technology and Research 3 (3) (2017) 70-82.


\bibitem[{yu}{yu}(2015)]{yu} J. Yuan, A.X. Liu, X. Ma, X.L. Liu, X. Qin, J.F. Zhang, The $g$-good-neighbor conditional diagnosability of $k$-ary $n$-cubes under the PMC model and MM* model, IEEE Transactions on Parallel and Distributed Systems 26 (2015) 1165--1177.



\bibitem[{ym}{ym}(2009)]{ym} W.H. Yang, J.X. Meng, Extraconnectivity of hypercubes, Applied Mathematics Letters 22 (6) (2009) 887-891.


\bibitem[{za}{za}(2016)]{za} S.R. Zhang, W.H. Yang, The $g$-extra conditional diagnosability and seequential $t/k$- diagnosability of hypercubes, International Journal of Computer Mathematics 93 (3) (2016) 482-497.

\bibitem[{Zhu}{Zhu}(2013)]{Zhu} Qiang Zhu, Xin-Ke Wang, Guanglan Cheng, Reliability evaluation of BC networks, IEEE Transactions on Computers 62 (11) (2013) 2337-2340.



\end{thebibliography}
\end{document}